\documentclass[11pt]{amsart}

\usepackage[all,cmtip]{xy}
\usepackage{amsmath, amssymb}
\usepackage{mathtools}
\usepackage[T1]{fontenc}
\usepackage[english]{babel}
\usepackage[hidelinks]{hyperref}
\usepackage[marginpar=2.5cm]{geometry}
\usepackage{xcolor}\usepackage{marginnote}
\usepackage{amsrefs}

\newtheorem{thm}{Theorem}[section]

\newtheorem*{thm*}{Theorem}
\newtheorem{lem}[thm]{Lemma}

\newtheorem{prop}[thm]{Proposition}
\newtheorem*{prop*}{Proposition}

\newtheorem{cor}[thm]{Corollary}

\theoremstyle{definition}
\newtheorem{defn}[thm]{Definition}
\newtheorem{notation}[thm]{Notation}
\newtheorem{remark}[thm]{Remark}
\newtheorem{question}[thm]{Question}
\newtheorem{example}[thm]{Example}

\def\la{\lambda}
\def\La{\Lambda}

\def\bb{\mathbb}

\def\de{\delta}

\def\bb{\mathbb}
\def\bd{\mathbf}

\def\sg{\sigma}
\def\G{\Gamma}
\def\cc{\mathcal}

\def\i{\imath}

\newcommand{\norm}[1]{\left\lVert #1 \right\rVert}

\DeclareMathOperator{\re}{Re}
\DeclareMathOperator{\im}{Im}
\DeclareMathOperator{\id}{id}

\DeclareMathOperator{\cp}{\mathcal{CP}}

\DeclareMathOperator{\tr}{tr}

\DeclareMathOperator{\Ch}{Ch}

\DeclareMathOperator{\Lin}{lin}
\DeclareMathOperator{\Quad}{quad}
\DeclareMathOperator{\diag}{diag}

\usepackage[mathscr]{euscript}

\newcommand\ip[2]{\left\langle #1\, , #2 \right\rangle}

\DeclareMathOperator{\CP}{\text{CP}}
\renewcommand{\P}{\text{P}}

\begin{document}

\title{Approximating Projections by Quantum Operations}

\author[Araiza]{Roy Araiza}

\address{Department of Mathematics \& IQUIST, University of Illinois at Urbana-Champaign, Urbana, IL 61801}
\email{raraiza@illinois.edu}
\urladdr{https://math.illinois.edu/directory/profile/raraiza}

\author[Griffin]{Colton Griffin}
\author[Khilnani]{Aneesh Khilnani}
\author[Sinclair]{Thomas Sinclair}

\address{Mathematics Department, Purdue University, 150 N. University Street, West Lafayette, IN 47907-2067}
\email{tsincla@purdue.edu}
\urladdr{http://www.math.purdue.edu/~tsincla/}

\subjclass[2020]{46L07; 05C50, 47C15, 81P47}

\keywords{quantum channels; graph invariants; semidefinite programming}


\setcounter{tocdepth}{1}
\maketitle

\begin{abstract}
    Using techniques from semidefinite programming, we study the problem of finding a closest quantum channel to the projection onto a matricial subsystem. We derive two invariants of matricial subsystems which are related to the quantum Lov\'asz theta function of Duan, Severini, and Winter.
\end{abstract}

\section{Introduction} 
A matricial system $\cc S$ will be a subspace of complex $n\times n$ matrices $M_n$ which contains the unit and is closed under taking adjoints, i.e., $\cc S^* = \cc S$. Matricial systems were first systematically studied by Choi and Effros \cite{ChoiEffros1977} and have been the subject of heavy investigation recently in quantum information theory under the guise of ``quantum graphs.'' Many interesting ``quantum'' extensions of classical graph invariants, such as clique, independence, and chromatic numbers and the Lov\'asz theta invariant, have been found for quantum graphs. (See, for example, \cite{Cameron2007, Winter2013}.) 

In the study of matricial systems, \emph{trace duality} plays an important role, which is the fact that the cone of positive-semidefinite matrices is self-polar under the bilinear pairing $(A, B)\mapsto \tr(B^*A)$. The trace pairing gives a non-degenerate inner product structure on $M_n$, known by various names such as the Frobenius or Hilbert--Schmidt inner product. A matricial system $\cc S\subset M_n$ can be completely characterized as the range of a unique orthogonal projection $P_{\cc S}$ with respect to the Hilbert--Schmidt inner product.

Given a linear operator $\Phi: M_n\to M_n$, we can define an adjoint $\Phi^\dagger: M_n\to M_n$ determined by the functional equation $\tr(B^*\Phi(A)) = \tr(\Phi^\dagger(B)^*A)$ for all $A,B\in M_n$. The projection $P_{\cc S}$ onto a matricial system has many nice properties. For instance, $P_{\cc S}$ is unital, i.e., $P_{\cc S}(1) = 1$, sends self-adjoint matrices to self-adjoint matrices, and has $P_{\cc S}^\dagger = P_{\cc S}$. The fact that $P_{\cc S}$ is unital and is its own adjoint implies that $P_{\cc S}$ is trace-preserving as well, that is, $\tr(P_{\cc S}(A)) = \tr(A)$ for all $A\in M_n$. However, $P_{\cc S}$ rarely preserves the set of positive-semidefinite matrices, so is generally not a \emph{quantum operation}, to which we refer the reader to the Background section below for a precise definition. In fact, $P_{\cc S}$ is a quantum operation exactly when $\cc S$ is injective as an operator system in the sense of Choi and Effros \cite{ChoiEffros1977}, the most significant examples being when $\cc S$ is a subalgebra of $M_n$. For a matricial system which is a subalgebra, $P_{\cc S}$ is known as the (unique) trace-preserving conditional expectation of $M_n$ onto $\cc S$. 

The goal of this note is to investigate quantum operations $\Phi$ which  approximate $P_{\cc S}$ as closely as possible while sharing the same properties outlined in the previous paragraph. Precisely, we would like to investigate quantum operations $\Phi: M_n\to M_n$ which are unital, preserve the trace, and whose range is equal to $\cc S$. As the projection canonically identifies the system, any such quantum operation should contain interesting information on the structure and properties of the matricial system and should be able to be used to derive useful invariants. As mentioned already, the best approximating operation to $P_{\cc S}$ being itself is equivalent to injectivity. 

We begin by defining an inner-product metric on the space of linear operators on $M_n$, compatible with the cone structure given by the positivity-preserving operators, which will be used as the metric for how closely a quantum operation approximates $P_{\cc S}$. Since we are working with an inner-product metric and the class of quantum operations we consider forms a non-empty convex set, this ensures that our problem is well-posed with a unique solution. From this basic problem, we define two numerical invariants $\phi_{\Quad}(\cc S)$, which is derived from the minimal distance quantum operation to $P_{\cc S}$, and $\phi_{\Lin}(\cc S)$ which measures the largest incident angle of such a quantum operation with $P_{\cc S}$. Crucially, $\phi_{\Lin}(\cc S)$ takes the form of a (complex) semidefinite program, so is effectively computable. Both of these invariants bear more than a passing resemblance to the quantum Lov\'asz theta invariant of Duan, Severini, and Winter \cite{Winter2013}, though both are distinct from it. We explicitly compute the optimal quantum operation for a family of matricial systems first studied by Farenick and Paulsen in \cite{FarenickPaulsen2012} and relate the answer to their work on quotients of operator systems.

After studying the problem in full generality, we concentrate on the case of matricial systems given by classical undirected graphs. In this case, we show that the approximation problem reduces from the quantum regime of operators on matrices to a simpler problem of matrix approximation in the Hilbert--Schmidt metric. As a consequence, we show in this case that $\phi_{\Quad}$ and $\phi_{\Lin}$ are both given by semidefinite programs. We discuss how these graph invariants are related to the famous and well-studied theta invariant of Lov\'asz \cite{Lovasz1979}. While we show that these invariants lack many of the properties that make the Lov\'asz theta invariant useful in so many applications, they may turn out to be of some interest in their own right. 

The outline of the paper is as follows. Section 2 contains background information on quantum operations and complex semidefinite programs. In Section 3 we define an inner product on the space of quantum operations and prove some basic properties about it. Section 4 contains the definitions of $\phi_{\Quad}$ and $\phi_{\Lin}$ along with proofs of all basic properties, computations, examples, and counterexamples.

\section{Background} 
We will denote the $n\times n$ complex matrices by $M_n$, and $M_n^+$ will denote the positive semidefinite $n\times n$ complex matrices. We will denote by $\cc L(M_n)$ the set of all linear maps from $M_n$ to itself. The trace of a matrix $A$ will be denoted by $\tr(A)$.

We define the Hilbert--Schmidt inner product on $M_n$ by 
\begin{equation*}
    \ip{A}{B}_2 := \tr(B^*A) = \sum_{i,j=1}^n A_{ij}\overline{B_{ij}},
\end{equation*} 
with $\|A\|_2 = \tr(A^*A)^{1/2}$ being the corresponding Hilbert--Schmidt norm on $M_n$. If $A, B\in M_n$ are hermitian, for convenience we will occasionally use $A\bullet B$ to denote the Hilbert--Schmidt inner product.

\begin{defn}
    We will say that a map $\Phi\in\cc L(M_n)$ is \emph{positive} if $\Phi(M_n^+)\subset M_n^+$. We will say that $\Phi$ is \emph{completely positive} if $\Phi\otimes \id_{M_k}\in \cc L(M_n\otimes M_k)$ is positive for all $k\in \bb N$. We will denote by $\cc P(M_n)$ the cone of positive maps in $\cc L(M_n)$ and by $\cp(M_n)$ the cone of completely positive maps.
\end{defn}

There is a canonical linear isomorphism $\Ch$ from $\cc L(M_n)$ to $M_n\otimes M_n$ given by 
\begin{equation*}
    \Ch: \Phi \mapsto \sum_{i,j=1}^n E_{ij}\otimes \Phi(E_{ij}).
\end{equation*}
The matrix $\Ch(\Phi)\in M_n\otimes M_n$ is known as the \emph{Choi matrix} associated to $\Phi$. The following is a foundational result of M.D.~Choi which we will use repeatedly; see \cite[Theorem 3.14]{paulsen2002completely} for a proof.

\begin{prop}[Choi's Theorem]
    The map $\Phi\in \cc L(M_n)$ is completely positive if and only if $\Ch(\Phi)$ is positive semidefinite.
\end{prop}

\subsection{Quantum operations}
We now recall the axiomatic approach to quantum operations. Fix a system $Q$ (a finite-dimensional Hilbert space) and let $S$ denote the set of density operators on $Q$. This is to say $\rho \in S$ if $\rho \in \cc L(Q)^+$ and $\tr(\rho) = 1.$ Then a \emph{quantum operation} is defined to be a map $\Phi: \cc L(Q) \to \cc L(Q)$ satisfying the following three properties: \begin{enumerate}
    \item $\tr(\Phi(\rho))$ is the probability that the process represented by the operation $\Phi$ occurs, when $\rho$ is the initial state of $Q$. Thus, $0 \leq \tr(\Phi (\rho)) \leq 1$ for all $\rho \in S$. 
    \item $\Phi$ is convex in the sense if $\{p_i\}_i$ is a probability distribution and $\{\rho_i\}_i \subset S$ then $\Phi(\sum p_i \rho_i) = \sum p_i \Phi(\rho_i).$
    \item $\Phi$ is completely positive. 
\end{enumerate}

Such examples of quantum operations are inner actions (conjugation by a unitary), and the partial trace. Throughout the manuscript we will consider a special class of quantum operations which will be linear maps $\Phi \in \cc L(M_n)$ such that $\Phi$ is unital completely positive and trace-preserving.

\subsection{Semidefinite programming}

Most of the literature on semidefinite programming focuses on the case of real matrices. As by necessity our semidefinite programs use hermitian matrices, we collect some background on complex semidefinite programs here. The reader may consult \cite{GoemansWilliamson2004, Watrous2009} for further information on complex semidefinite programs and \cite{GartnerMatousek2012, Lovasz2003} for the general theory of semidefinite programming.

As our starting point we will say that a \emph{complex semidefinite program} is an optimization problem which can be expressed in the following form:

\begin{equation}\label{eq:complex-primal}
    \begin{aligned}
        &\text{maximize} && C\bullet X\\
        &\text{subject to} && A_i\bullet X = b_i,\ i=1, \dotsc, k\\
        & && X\in M_n^+
    \end{aligned}
\end{equation}
where $C$ and $A_1,\dotsc,A_k$ are hermitian $n\times n$ complex matrices and $b_1,\dotsc,b_k$ are (necessarily) real numbers.

As observed in \cite[section 3]{GoemansWilliamson2004}, every complex semidefinite program of the form (\ref{eq:complex-primal}) can be written as a real semidefinite program as follows.
    
\begin{equation}\label{eq:complex-to-real}
    \begin{aligned}
        &\text{maximize} && C'\bullet Y\\
        &\text{subject to} && A_i'\bullet Y = 2b_i, &&& i=1, \dotsc, k\\
        & && F_{ij}\bullet Y =0, &&&i,j=1,\dotsc, n\\
        & && G_{ij}\bullet Y =0, &&&i,j=1,\dotsc,n\\
        & && Y\in M_{2n}(\bb R)^+
    \end{aligned}
\end{equation}
where for $A\in M_n$ we write 
\begin{equation*}
    \begin{split}
        &A' := \begin{pmatrix} \re A & -\im A\\ \im A & \re A  \end{pmatrix},\ C' := \begin{pmatrix} \re C & - \im C\\ \im C & \re C\end{pmatrix}\\
        &F_{ij} := \begin{pmatrix} E_{ij} & 0\\ 0 & -E_{ij}\end{pmatrix},\ G_{ij} := \begin{pmatrix} 0 & E_{ij}\\ E_{ij} & 0\end{pmatrix}.
    \end{split}
\end{equation*}
Note that $\re A$ and $\im A$ represent the element-wise real and imaginary parts, i.e. $(\re A)_{ij} = \re(A_{ij})$.
It is straightforward to check that the program (\ref{eq:complex-to-real}) returns twice the value of the program (\ref{eq:complex-primal}).

We now write the dual program to (\ref{eq:complex-to-real}):

\begin{equation}\label{eq:complex-to-real-dual}
    \begin{aligned}
        &\text{minimize} && 2b^Ty\\
        &\text{subject to} && \sum_{i=1}^k y_i A_i' + Z - C'\in M_{2n}(\bb R)^+
    \end{aligned}
\end{equation}
where $b = [b_1,\dotsc,b_k]^T$ and $Z = \begin{pmatrix} P & Q\\ Q & -P\end{pmatrix}$ for some $P,Q\in M_n(\bb R)$. Since $A_1',\dotsc,A_k',C'$ are invariant under the involution $\begin{pmatrix} A & B\\ C & D\end{pmatrix}\mapsto \begin{pmatrix} D & -C\\ -B & A\end{pmatrix}$ and this involution preserves $M_{2n}(\bb R)^+$ and leaves the objective function invariant, by averaging we can omit the term $Z$ without affecting the value of the program. Translating the resulting program back to complex form gives the complex dual program to the complex primal semidefinite program (\ref{eq:complex-primal}):

\begin{equation}\label{eq:complex-dual}
    \begin{aligned}
        &\text{minimize} &&b^Ty\\
        &\text{subject to} && \sum_{i=1}^k y_iA_i - C\in M_n^+
    \end{aligned}
\end{equation}
where $b = [b_1,\dotsc,b_k]^T$.

This allows us to restate the strong duality theorem for (real) semidefinite programs in the setting of complex semidefinite programs. We thus refer the reader to \cite[chapter 4]{GartnerMatousek2012} or \cite[section 6.3.1]{Lovasz2003} for background on duality theory in semidefinite programming and a proof of the following result.

\begin{prop}
    Consider the optimal values $v_{primal}$ and $v_{dual}$ of the programs (\ref{eq:complex-primal}) and (\ref{eq:complex-dual}), respectively. If both programs are feasible, we have that $v_{dual}\geq v_{primal}$. Moreover, if there is $X\in M_n^+$ invertible so that $A_i\bullet X = b_i$ for all $i=1,\dotsc,k$, then $v_{primal} = v_{dual}$.
\end{prop}


\section{An Inner Product on the Space of Quantum Operations}

\begin{defn}\label{def:F-norm} We define an inner product on $\cc L(M_n)$ by 
\begin{equation}\label{def:K-inner-product}
    \begin{split}
        \ip{\Phi}{\Psi}_K :&= \sum_{i,j} \tr(\Phi(E_{ii})\Psi(E_{jj})^*) + \sum_{i\not=j} \tr(\Phi(E_{ij})\Psi(E_{ij})^*)\\
        &= \tr(\Phi(I_n)\Psi(I_n)^*) + \sum_{i\not=j} \tr(\Phi(E_{ij})\Psi(E_{ij})^*).
    \end{split}
\end{equation}
\end{defn}

For two maps $\Phi, \Psi: M_n\to M_n$ we write $\Phi\prec\prec \Psi$ if $\Phi(x)\preceq\Psi(x)$ for all $x$ positive semidefinite, i.e., if $\Psi - \Phi$ is a positive map.

\begin{lem}\label{lem:F-decreasing} Let $\Phi,\Psi\in \cc P(M_n)$ be positive. If $\Phi\prec\prec \Psi$, then $\|\Phi\|_K\leq \|\Psi\|_K$.
\end{lem}

\begin{proof}
 Let $x,y\succeq 0$ be positive matrices. Since $\Phi\prec\prec \Psi$ we have that $\tr(\Phi(x)y)\leq \tr(\Psi(x)y)$. Now using that $\Phi(x)$ and $\Psi(x)$ are positive with $\Phi(x)\preceq \Psi(x)$, we have that \[\tr(\Phi(x)\Phi(x))\leq \tr(\Phi(x)\Psi(x))\leq \tr(\Psi(x)\Psi(x)).\] 
 
Let $\bd v = (v_1,\dotsc,v_n)$ be a random vector where each entry $v_i\in \bb T$ is chosen independently is distributed uniformly with regard to Lebesgue (probability) measure. In this way $\bd x := \bd v\otimes \bar{\bd v}$ is a random positive semidefinite matrix. Observing
    \begin{equation}
        \bb E_{\bd v} v_iv_j\bar v_k \bar v_l = \max\{\de_{i,k}\de_{j,l}, \de_{i,l}\de_{j,k}\},
    \end{equation}
    we have that 
    \begin{equation}\label{eq:F-norm-expectation}
        \begin{aligned}
            \bb E_{\bd x} \tr(\Phi(\bd x)\Psi(\bd x)^*) &=\bb E_{\bd v} \tr(\Phi(\bd v \otimes \bar{\bd v})\Psi(\bd v \otimes \bar{\bd v})^*)\\
            &= \bb E_{\bd v}\sum_{i,j,k,l} v_i\bar v_j\bar v_k v_l\tr(\Phi(E_{ij})\Psi(E_{kl})^*)\\
            &= \sum_{i,j,k,l} \bb E_{\bd v}(v_i\bar v_j\bar v_k v_l)\tr(\Phi(E_{ij})\Psi(E_{kl})^*)\\
            &= \sum_{i,k} \tr(\Phi(E_{ii})\Psi(E_{kk})^*)
            + \sum_{i\not=j} \tr(\Phi(E_{ij})\Psi(E_{ij})^*)= \ip{\Phi}{\Psi}_K.
        \end{aligned}
    \end{equation}
It follows that $\|\Phi\|_K\leq \|\Psi\|_K$ by averaging. \qedhere
\end{proof}

\begin{notation}
    For $A,B\in M_n$, let $A\circ B$ denote the entrywise (Schur) product of the matrices, i.e., $(A\circ B)_{ij} = A_{ij}B_{ij}$.
\end{notation}

\begin{remark}\label{rmk:schur}
    If $\Phi(X) = A\circ X$ and $\Psi(X) = B\circ X$, then $\Phi(E_{ii})\Psi(E_{jj})=0$ unless $i=j$. Hence,
\begin{equation}\label{F-norm-hilbert-schmidt}
    \ip{\Phi}{\Psi}_K = \sum_i A_{ii}\overline{ B_{ii}} + \sum_{i\not=j} A_{ij}\overline{ B_{ij}} =\tr(B^*A) = \ip{A}{B}_2.
\end{equation}
If we have that $\Phi$ and $\Psi$ are unital, then we have that 
\begin{equation}\label{eq:K-ip-unital}
    \ip{\Phi}{\Psi}_K = n + \sum_{i\not= j} \tr(\Phi(E_{ij})\Psi(E_{ij})^*).
\end{equation}
\end{remark}

\begin{defn}\label{defn:Phi-prime}
    For $\Phi\in \cc L(M_n)$, we define $\Phi'\in \cc L(M_n)$ by 
    \begin{equation}
        \begin{aligned}
            & \Phi'(E_{ii}) = \Phi(E_{ii})\circ I_n && i=1,\dotsc,n\\
            & \Phi'(E_{ij}) := \Phi(E_{ij})_{ij} E_{ij} && \text{if}\ i\not= j.
        \end{aligned}
    \end{equation}
\end{defn}

\begin{lem}\label{lem:F-norm-averaging}
    We have that $\|\Phi'\|_K\leq \|\Phi\|_K$. Moreover, $\Phi'$ is positive, completely positive, unital, or trace-preserving if $\Phi$ is.
\end{lem}

\begin{proof}
    Let $\bd v$ be the random vector as in the proof of Lemma \ref{lem:F-decreasing}. If Q is either a diagonal unitary matrix or a permutation matrix, we see that $Q\bd v$ and $\bd v$ are identically distributed random vectors. Setting $\Phi^Q(x) := Q^*\Phi(QxQ^*)Q$, we see from equation (\ref{eq:F-norm-expectation}) that
    \begin{equation}\label{eq:F-norm-invariance}
        \|\Phi^Q\|_K = \|\Phi\|_K.
    \end{equation}
    Let $\cc U_n$ be the group of all unitary diagonal matrices in $M_n$ equipped with Haar (probability) measure. For $\bd U := \diag(\bd v)\in \cc U_n$ a Haar-uniformly distributed random variable we have that, similarly to equation (\ref{eq:F-norm-expectation}),
    \begin{equation}\label{eq:prime-averaging}
        \bb E_{\bd U}\, \Phi^{\bd U} = \Phi'.
    \end{equation}
    Indeed,
    \begin{equation}
        \bb E_{\bd U}\, U^*\Phi(UE_{ij}U^*)U = \bb E_{\bd v}\, v_i\bar v_j\sum_{k,l}\bar v_k v_l\Phi(E_{ij})_{kl} E_{kl}.
    \end{equation}
    From this it follows that if $i\not= j$, then $i=k$ and $j=l$. If $i=j$, then the sum reduces to summing over all $k=l$.
    It then follows by convexity of the norm that $\|\Phi'\|_K\leq \|\Phi\|_K$.
\end{proof}

The following lemma is a variant of the previous, and seems well-known. For the sake of convenience we reproduce here the treatment in \cite{AGS}.

\begin{lem}\label{lem:choi-expectation}
    For $\Phi\in\cp(M_n)$, a completely positive map, the matrix $A_\Phi$ defined by $[A_\Phi]_{ij} := \Phi(E_{ij})_{ij}$ is positive semidefinite. Moreover, $\max_i [A_\Phi]_{ii} \leq \|\Phi(1)\|$. 
\end{lem}

\begin{proof}
    The map $\Delta: E_{ij}\mapsto E_{ij}\otimes E_{ij}$ induces a (non-unital) $\ast$-embedding of $M_n$ into $M_n\otimes M_n$. This implies that $\Delta(B) = \sum_{ij} B_{ij}E_{ij}\otimes E_{ij}$ is positive semidefinite for all $B\in M_n$ positive semidefinite. Letting $J_n$ denote the $n\times n$ matrix with all entries $1$, we see that
    \begin{equation*}
        \tr(A_{\Phi}B) = \tr((\Ch(\Phi)\circ \Delta(B))\Delta(J_n)) \geq 0
    \end{equation*}
    for all $B\in M_n$ positive semidefinite; thus, $A_\Phi$ is positive semidefinite. The second assertion follows since $(A_\Phi)_{ii} = \Phi(E_{ii})_{ii} \leq \|\Phi(E_{ii})\|\leq \|\Phi(1)\|$. \qedhere
\end{proof}

\begin{cor}\label{cor:improved-expectation}
    For $\Phi\in\cp(M_n)$ define the matrix $B_\Phi$ by 
    \begin{equation*}
        \begin{aligned}
            & [B_\Phi]_{ii} = \|\Phi(1)\| && i=1,\dotsc, n\\
            & [B_\Phi]_{ij} = \Phi(E_{ij})_{ij} && i\not= j.
        \end{aligned}
    \end{equation*}
    We have that $B_\Phi$ is positive semidefinite.
\end{cor}

\begin{proof}
    By the proof of Lemma \ref{lem:choi-expectation}, we have that $A_\Phi$ is positive semidefinite, and $B_\Phi-A_\Phi$ is a diagonal matrix with non-negative entries, and the result follows. \qedhere
\end{proof}

\begin{cor}\label{cor:norm-ineq}
    Suppose that $\Phi,\Psi\in\cc L(M_n)$ are both completely positive and unital. We have that 
    \begin{equation}
        \|B_\Phi - B_\Psi\|_2\leq \|\Phi - \Psi\|_K.
    \end{equation}
\end{cor}

\begin{proof}
    Since $\Phi,\Psi\in\cc L(M_n)$ are unital, we have that $[B_\Phi - B_\Psi]_{ii} = 0$. By (\ref{F-norm-hilbert-schmidt}) we have under these assumptions that $\|\Phi' - \Psi'\|_K = \|B_\Phi-B_\Psi\|_2$. The result then follows by Lemma \ref{lem:F-norm-averaging}.
\end{proof}

We end this section with one more observation on the $\|\cdot\|_K$-norm, which will not be used in the sequel.

Let $\Phi: M_n\to M_n$ be a map which is self adjoint in the sense that $\Phi(x^*) = \Phi(x)^*$. Define 
\begin{equation*}
    \cp(\Phi) := \{\Psi\in \cp(M_n) : \Phi\prec\prec \Psi\}
\end{equation*}
Note that $\cp(\Phi)$ is a non-empty convex set. By Zorn's lemma and closedness, the set $\cp(\Phi)$ has at least one $\prec\prec$-minimal element.

\begin{prop}
Let $\Phi\in \cc L(M_n)$ be self adjoint. Let $\Psi_\ast$ be the unique element of $\cp(\Phi)$ minimizing $\Psi\mapsto \|\Psi - \Phi\|_K$. Then $\Psi_\ast$ is $\prec\prec$-minimal in $\cp(\Phi)$.
\end{prop}

\begin{proof}
 Suppose there is $\La\in \cp(\Phi)$ so that $\La\prec\prec \Psi_{\ast}$ We have that $\Psi_\ast - \Phi$ and $\La - \Phi$ are positive and $\La - \Phi\prec\prec \Psi_\ast - \Phi$, hence by Lemma \ref{lem:F-decreasing} we have that $\|\La - \Phi\|_K\leq \|\Psi_\ast - \Phi\|_K$, hence $\La = \Psi_\ast$ by the minimality of $\Psi_\ast$. \qedhere
\end{proof}

\section{Approximating Quantum Operations}

\begin{defn}
    We say that subspace $\cc S\in M_n$ is a \emph{matricial system} if it contains the unit and is closed under taking adjoints.
\end{defn}

Given some matricial system $\mathcal{S}$, we consider the orthogonal projection $P_{\cc S}: M_n\to \cc S$ with respect to the Hilbert--Schmidt inner product. It holds for any matricial system that $P_{\cc S}\in \cc L(M_n)$ is unital and trace-preserving; however, it is rarely the case the $P_{\cc S}$ is positive, let alone completely positive. We seek to approximate $P_\mathcal{S}$ by a quantum operation with certain properties. These properties should be shared with $P_{\cc S}$ and allow us to use
Lemma \ref{lem:F-norm-averaging}, and so we establish our program to be
\begin{equation}\label{eq:min-quantum-operation}
\begin{aligned}
    \Phi_*=\operatorname{argmin}&&&\| \Phi-P_\mathcal{S}\|_K\\
    \text{s.t.}&&&\Phi(M_n)\subset \mathcal{S},\\
    &&&\Phi\textrm{ completely positive},\\
    &&&\Phi\textrm{ unital},\\
    &&&\Phi\textrm{ trace preserving}.
\end{aligned}
\end{equation}
Alternately, we can cast these conditions in terms of the Choi matrix $\Ch(\Phi)$:
\begin{equation}\label{eq:PhiProgram}
\begin{aligned}
    \Phi_*=\operatorname{argmin}&&&\| \Phi-P_\mathcal{S}\|_K\\
    \text{s.t.}&&&\Ch(\Phi)\in M_n\otimes \mathcal{S},\\
    &&&\Ch(\Phi)\in (M_n\otimes M_n)^+,\\
    &&&\tr\otimes\id(\Ch(\Phi))=I_n,\\
    &&&\id\otimes\tr(\Ch(\Phi))=I_n.
\end{aligned}
\end{equation}
An alternative to this objective is to maximize the inner product $\max\langle\Phi,P_\mathcal{S}\rangle_{K}$. Geometrically, inner products measure angles, and so maximizing this inner product can be achieved by reducing the angle between $\Phi$ and $P_\mathcal{S}$.

\begin{defn}
Define $\phi_{\Quad}(\mathcal{S}):= \frac{1}{n}\langle \Phi_*,P_\mathcal{S}\rangle_K$, where $\Phi_*$ is given by the program (\ref{eq:PhiProgram}). Define $\phi_{\Lin}(\mathcal{S})$ to be the same function with the alternative objective function
\begin{equation}\label{eq:PhiLin}
\begin{aligned}
    \phi_{\Lin}(\mathcal{S}):=\max&&&\frac{1}{n}\langle\Phi,P_\mathcal{S}\rangle_K\\
    \text{s.t.}&&&\Ch(\Phi)\in M_n\otimes \mathcal{S},\\
    &&&\Ch(\Phi)\in (M_n\otimes M_n)^+,\\
    &&&\tr\otimes\id(\Ch(\Phi))=I_n,\\
    &&&\id\otimes\tr(\Ch(\Phi))=I_n.
\end{aligned}
\end{equation}
\end{defn}

\begin{prop}\label{eq:phi_linquad_ineq}
For any matricial system $\cc{S}\subset M_n$,
\begin{equation}
    \phi_{\Lin}(\cc{S})\geq \phi_{\Quad}(\cc{S}).
\end{equation}
\end{prop}

\begin{proof}
The argument $\Phi_*$ from the program (\ref{eq:PhiProgram}) satisfies all the constraints of the program (\ref{eq:PhiLin}) for $\phi_{\Lin}(\cc{S})$. Thus the value $\phi_{\Quad}(\cc{S})=\frac{1}{n}\langle\Phi_*,P_\mathcal{S} \rangle_K$ is a lower bound for the value $\phi_{\Lin}(\cc{S})$.
\end{proof}


\begin{example} 
    Consider the matricial system of all $n\times n$ matrices with constant diagonal entries, 
    \begin{equation}\label{eq:diag-system}
        \cc S_n=\{X\in M_n : X_{ii}=X_{jj},\, \forall i,j=1,\ldots,n\}.
    \end{equation}
    We see that $P_n$, the orthogonal projection $\cc S_n$, is given by  $X\mapsto \tilde{X}$ where $\tilde{X}_{ij}=X_{ij}$ for $i\neq j$, but $\tilde{X}_{ii}=\tr(X)/n$. From this, we can see that the Choi matrix of $P_n$ is
    $$\Ch(P_n)= \frac{1}{n}I_{n^2}+\sum_{i\neq j}E_{ij}\otimes E_{ij}.$$

    We can view the case of $\cc S_2\subset M_2$ via the Pauli matrices where $M_2 = {\rm span}\{1, \sg_x, \sg_y, \sg_z\}$ and $\cc S_2 = {\rm span}\{1, \sg_x,\sg_y\}$. In this instance $P_2(1)=1$, $P_2(\sg_x) = \sg_x$, $P_2(\sg_y) = \sg_y$, and $P_2(\sg_z)=0$. It can be seen that the projection $P_2$ is positive, but not $2$-positive. There is a minimal unital completely positive map which sits over this projection, in the sense that the difference is positivity preserving, and it is the one that maps $(1,\sg_x,\sg_y,\sg_z)\mapsto (2, \sg_x, \sg_y, 0)$.

\end{example}

The optimal quantum operation for the program (\ref{eq:PhiProgram}) can be computed exactly.
%

\begin{prop}
    Given the matricial system $\cc S_n$ from (\ref{eq:diag-system}), the optimal quantum operation $\Phi_*$ is given by
    \begin{equation}\label{eq:diag-optimal}
        \Phi_*(A)=\frac{\tr(A)}{n}I_{n} + \sum_{i\neq j}\frac{A_{ij}}{n}E_{ij}.
    \end{equation}
\end{prop}

\begin{proof}
    Applying Definition \ref{defn:Phi-prime}, we see that $P_n = P_n'$; hence, by Lemma \ref{lem:F-norm-averaging}, we can assume that $\Phi_\ast = \Phi_\ast'$ since $\|\Phi_\ast' - P_n\|_K \leq \|\Phi_\ast - P_n\|_K$. Moreover, since $\Phi_\ast(M_n)\subset \cc S_n$, it follows that $\Phi_\ast'(M_n)\subset \cc S_n$ using the identity (\ref{eq:prime-averaging}) and the fact that $Q^*\cc S_n Q = \cc S_n$ where $Q$ is any diagonal unitary matrix. Since $\Phi_\ast'(E_{ii})$ is a diagonal matrix for all $i=1,\dotsc, n$ and the only diagonal matrices in $\cc S_n$ are scalar multiples of the identity, it follows that $\Phi_\ast'(E_{ii}) = \frac{1}{n}I_n$ since $\Phi_\ast'$ is trace-preserving.
    
    For ease of notation, for $i\not= j$, let $B_{ij} := \Phi_\ast'(E_{ij})_{ij}$. We have for all $|\la|\leq 1$ that 
    \[E_{ii}\otimes E_{11} -\la E_{ij}\otimes E_{12} -\bar\la E_{ji}\otimes E_{21} + E_{jj}\otimes E_{22}\succeq 0.\]
    Since $\Phi_\ast'$ is unital, completely positive again by Lemma \ref{lem:F-norm-averaging}, this implies that
    \begin{equation*}
        \Phi_\ast'(E_{ii})\otimes E_{11} - \la\Phi_\ast'(E_{ij})\otimes E_{12} - \bar\la\Phi_\ast'(E_{ji})\otimes E_{21} + \Phi_\ast'(E_{jj})\otimes E_{22}\succeq 0;
    \end{equation*}
    hence, 
    \begin{equation*}
        \frac{1}{n}I_n\otimes E_{11} - |B_{ij}| E_{ij}\otimes E_{12} - |B_{ij}| E_{ji}\otimes E_{21} + \frac{1}{n} I_n\otimes E_{22}\succeq 0.
    \end{equation*}
    This implies that $|B_{ij}|\leq \frac{1}{n}$.
    
    It is now easy to see that the distance is minimized when $B_{ij} = \frac{1}{n}$ for all $i\not= j$. \qedhere
\end{proof}

\begin{remark}
    Let $\cc J_n\subset M_n$, be the subspace of all diagonal matrices of trace zero. We see that $\cc S_n = \cc J_n^\perp$. There is a canonical operator space structure on $M_n/\cc J_n$, first studied in detail by Farenick and Paulsen \cite{FarenickPaulsen2012}. It is shown therein that the map $\Phi_\ast$ given in (\ref{eq:diag-optimal}) gives a complete isometry $\Phi_\ast: M_n/\cc J_n\to \cc S_n$. 
    
    Defining an operator system structure on $\cc S_n$ by $X\in M_k(\cc S_n)$ is positive if 
    \begin{equation*}
        \Phi_\ast^{-1}\otimes \id_{M_k}(X)\cap M_{nk}^+\not= \emptyset,
    \end{equation*}
    we observe that $\cc S_n$ equipped with this operator system structure is completely order isomorphic to $M_n/\cc J_n$.
\end{remark}

\begin{question}
    For any matricial system $\cc S\subset M_n$, is it true that the minimizing map $\Phi_\ast$ in (\ref{eq:min-quantum-operation}) gives a complete isometry $\Phi_\ast: M_n/\cc S^{\perp}\to \cc S$?
\end{question}

\subsection{Applications to Graph Systems}

\begin{defn}
For a given graph $\G=(V,E)$ with $|V|=n$, we define the corresponding \emph{graph system} to be
\begin{equation}\label{graph_subsystem}
    \cc S_{\G}:=\{X\in M_n(\mathbb{C}) : X_{ij}=0,\; i\neq j,\; (i,j)\notin E\}.
\end{equation}
\end{defn}
The corresponding projection for this system is 
\[P_{\cc S_{\G}}=\sum_i E_{ii}\otimes E_{ii} + \sum_{(i,j)\in E}E_{ij}\otimes E_{ij}.\] The functions $\phi_{\Lin}(\cc S_\G)$ and $\phi_{\Quad}(\cc S_\G)$ will also be written equivalently as
$\phi_{\Lin}(\G)$ and $\phi_{\Quad}(\G)$.

\begin{notation}
We use the definition of the strong graph product $\G\boxtimes\La$ of two graphs $\G$ and $\La$ given in \cite{GartnerMatousek2012}*{Definition 3.4.1}. Equivalently, we have that the strong product may be defined by the relation
\begin{equation*}
    \cc S_{\G}\otimes \cc S_\La = \cc S_{\G\boxtimes\La}.
\end{equation*}
\end{notation}

\begin{notation}
For a given graph $\G=(V,E)$, we define the graph complement $\overline \G = (V, \overline E)$ where $(i,j)\in \overline E$ if and only if $i\not= j$ and $(i,j)\not\in E$. We recall that the clique number $\omega(\Gamma)$ is defined as the size of the maximal subset of vertices such that all vertices are connected. Similarly, the independence number $\alpha(\G)=\omega(\overline{\G})$ is defined as the size of the maximal subset of vertices such that no vertices are connected.
\end{notation}

We observe that graph systems may be characterized as those matricial systems $\cc S$ for which the orthogonal project $P_{\cc S}$ is a Schur multiplier.

\begin{lem}\label{lem:schur-replace}
    Let $P_\mathcal{S}$ be a Schur multiplier. Then the program (\ref{eq:PhiProgram}) is minimized by $\Phi$ being a Schur multiplier.
\end{lem}

\begin{proof}
    From Lemma \ref{lem:F-norm-averaging} we know that given a selected $\Phi$ we have that for $\Phi'$ as in Definition \ref{defn:Phi-prime} that $\| \Phi'-P_{\cc S}'\|_K=\| \Phi'-P_{\cc S}\|_K\leq \| \Phi-P_{\cc S}\|_K$. However, $\Phi'$ is not necessarily a Schur multiplier. From Corollary \ref{cor:improved-expectation}, consider the positive semidefinite matrix $B_\Phi$, and denote by $\tilde\Phi$ its associated Schur multiplier. We have from that result that $\Phi'(E_{ij}) = \tilde\Phi(E_{ij})$ for all $i\not= j$, and that $\tilde\Phi$ is unital, trace-preserving and completely positive.

    To finish, we note that $\|\tilde{\Phi}-P_{\cc S}\|_K = \|\Phi'-P_{\cc S}\|_K$. This follows from $\tilde\Phi$, $\Phi'$, and $P_{\cc S}$ being unital, so $(\tilde\Phi - P_{\cc S})(I_n) = 0 = (\Phi' - P_{\cc S})(I_n)$. Equality follows by Definition \ref{def:F-norm}.\qedhere
    
\end{proof}

\begin{cor}\label{cor:phi-lin-graph}
Given a graph system $S_\G$ for $\G=(V,E)$, the functions $\phi_{\Quad}(\G)$ and $\phi_{\Lin}(\G)$ are given by the semidefinite programs
\begin{equation}\label{eq:PhiLinGraph}
\begin{aligned}
    \phi_{\Lin}(\Gamma)=\max&&&(A\bullet J_n)/n,\\
    \textup{s.t.}&&&A_{ii}=1,\:\forall i=1,\ldots, n,\\
    &&&A_{ij}=0\textup{ if }(i,j)\in\overline{E},\\
    &&&A\succeq 0
\end{aligned}
\end{equation}
and $\phi_{\Quad}(\G)=(A_\ast\bullet J_n)/n$, where
\begin{equation}\label{eq:PhiQuadGraph}
\begin{aligned}
    A_\ast=\operatorname{argmin}&&&\| A- J_n\|_2\\
    \textup{s.t.}&&&A_{ii}=1,\:\forall i=1,\ldots, n,\\
    &&&A_{ij}=0\textup{ if }(i,j)\in\overline{E},\\
    &&&A\succeq 0,
\end{aligned}
\end{equation}
where $J_n\in M_n$ is the matrix of all ones.
\end{cor}

\begin{proof}
    Let $P_\G := I_n + \sum_{(i,j)\in E} E_{ij}\in M_n$ be the augmented adjacency matrix of $\G$. 
    
    We begin by considering the program (\ref{eq:PhiLin}), the objective for which is $\frac{1}{n}\ip{\Phi}{P_{\cc S_\G}}_K$.
    Since $P_{\cc S_\G}$ is a Schur multiplier, we have by (\ref{eq:K-ip-unital}), using the same notation as in Corollary \ref{cor:improved-expectation}, that
    \begin{equation*}
        \frac{1}{n}\ip{\Phi}{P_{\cc S_\G}}_K = \frac{1}{n}B_\Phi\bullet P_\G = \frac{1}{n} B_\Phi\bullet J_n.
    \end{equation*}
    The last equality follows by noting that $\Phi(M_n)\subset \cc S$; hence, we have that for all $i\not=j$ that $(B_\Phi)_{ij} =\Phi(E_{ij})_{ij} = 0$ if $(i,j)\not\in E$. The Schur multiplier associated to $B_\Phi$ is unital, completely positive and trace-preserving; thus, $\Phi$ may be taken to be a Schur multiplier associated to $B_\Phi$.
    
    We now turn our attention to the program (\ref{eq:min-quantum-operation}).
    By the proof of Lemma \ref{lem:schur-replace} we may replace $\Phi_\ast$ by the Schur multiplier associated to $B_{\Phi_\ast}$.  By Remark \ref{rmk:schur} we have that 
    \begin{equation*}
        \|\Phi_\ast - P_{\cc S_\G}\|_K = \|B_{\Phi_\ast} - P_\G\|_2.
    \end{equation*}
    Since we have already noted that for any $\Phi$ satisfying the constraints of (\ref{eq:min-quantum-operation}) that $B_\Phi\bullet P_\G = B_\Phi\bullet J_n$, we have that the squared objective $\|B_{\Phi} - P_\G\|_2^2$ is up to a constant independent of $\Phi$ equal to $\|B_{\Phi} - J_n\|_2^2$. Thus, these objectives are interchangeable when computing argmin, and the result is obtained. \qedhere
\end{proof}

We recall that from \cite[Theorem 4]{Lovasz1979} that the \emph{Lov\'asz theta function} $\vartheta(\G)$ of a graph $\G = (V, E)$ can be expressed as the following semidefinite program:

\begin{equation}\label{eq:lovasz-sdp}
\begin{aligned}
    \vartheta(\Gamma)=\max&\quad Y\bullet J_n\\
    \text{s.t.} 
    &\quad Y_{ij}=0\text{ if } (i,j)\in E,\\
    &\quad \tr(Y)=1,\\
    &\quad Y\succeq 0.
\end{aligned}
\end{equation}

\begin{prop}\label{prop:phi-quad-omega}
For any graph $\Gamma=(V,E)$,
\begin{equation}
    \vartheta(\Gamma)\geq\phi_{\Lin}(\overline{\Gamma})\geq \phi_{\Quad}(\overline{\Gamma}).
\end{equation}
\end{prop}
\begin{proof}
We make the substitution $Y=A/n$ into the program (\ref{eq:lovasz-sdp}) to obtain:
\begin{equation}\label{eq:lovasz-sdp-var}
\begin{aligned}
    \vartheta(\Gamma)=\max&\quad (A\bullet J_n)/n\\
    \text{s.t.}
    &\quad A_{ij}=0\text{ if } (i,j)\in E,\\
    &\quad \tr(A)=n,\\
    &\quad A\succeq 0.
\end{aligned}
\end{equation}
By Corollary \ref{cor:phi-lin-graph} the optimal argument $A_\ast$ from $\phi_{\Lin}(\overline{\Gamma})$ satisfies the constraints of $\vartheta(\Gamma)$ since $\tr(A_\ast)=\sum_{i}(A_\ast)_{ii}=\sum_{i}1=n$. The optimal value of the semidefinite program (\ref{eq:lovasz-sdp-var}) is at least this value, so $\vartheta(\Gamma)\geq\phi_{\Lin}(\overline{\Gamma})$. The second inequality follows directly from Proposition \ref{eq:phi_linquad_ineq}.
\end{proof}

Although the inequality $\phi_{\Lin}(\Gamma)\geq \omega(\Gamma)$ does not hold (see the second inequality in Proposition \ref{thingsthatarewrong}), a slight relaxation of the right side does work.
\begin{prop}
For any graph $\Gamma=(V,E)$ with $|V|=n$,
\begin{equation}\label{eq:phi-ineq}
    \phi_{\Quad}(\Gamma)\geq \omega(\Gamma)(\omega(\Gamma)-1)/n+1.
\end{equation}
\end{prop}
\begin{proof}
If we have some maximal clique set $C$ where $|C|=\omega(\Gamma)$, then the number of edges is $|C|(|C|-1)/2$. Now order the graph so that $1,\ldots, |C|$ are the numbers corresponding to the elements of $C$. Let our $n\times n$ matrix be
\begin{equation*}
    A = \begin{bmatrix}
    J_{|C|}&0\\
    0&I_{n-|C|}
    \end{bmatrix}.
\end{equation*}
All of the constraints hold since the diagonal is all ones, we have zeros wherever we don't have an edge, and the block matrix is positive semidefinite since $J_n$ and $I_n$ are for any $n$. Therefore, the optimal value of the program must be at least $(J\bullet A)/n=\omega(\Gamma)(\omega(\Gamma)-1)/n+1$.
\end{proof}

\begin{prop}
The program \eqref{eq:PhiQuadGraph} in the definition of $\phi_{\Quad}(\Gamma)$ can be written as the semidefinite program
\begin{equation}\label{eq:quadphi}
\begin{aligned}
    A=\operatorname{argmin}&&&t\\
    \text{s.t.}&&&A_{ii}=1,\:\forall i=1,\ldots, n,\\
    &&&A_{ij}=0\textup{ if }(i,j)\in\overline{E},\\
    &&&\begin{bmatrix}
A&0&0\\
0&I&\vec{A}\\
0&\vec{A}^T&2A\bullet J_n+t
\end{bmatrix}\succeq 0.
\end{aligned}
\end{equation}
\end{prop}

\begin{proof}
 Consider the augmented adjacency matrix $P = P_\G := I_n + \sum_{(i,j) \in E} E_{ij}$. Every $A \in M_n$ which satisfies the constraints of the above program satisfies $\tr(AP) = \tr(AJ_n)$ and therefore we may interchange $P$ and $J_n$ in our calculations. Given such an $A \in M_n$, consider the vector $\vec{A} \in \bb R^{n^2}$ defined as \begin{equation}
    \vec{A}=(A_{11},\ldots,A_{n1},A_{12},\ldots,A_{n2},\ldots, A_{1n},\ldots, A_{nn}).
\end{equation} The objective function in \eqref{eq:PhiQuadGraph} is defined as \begin{align*}
    \norm{A-P}_2^2 = \tr((A-P)^2) = A\bullet A - 2A\bullet P + P\bullet P=\vec{A}^T\vec{A}-2A\bullet P + P\bullet P.
\end{align*}Using Schur complements \cite{Zhang2005}*{Thm 1.12}, we can say that
$$
\vec{A}^T\vec{A}-2A\bullet P + P\bullet P - t\leq 0\iff
\begin{bmatrix}
I&\vec{A}\\
\vec{A}^T&2A\bullet P-P\bullet P+t
\end{bmatrix}\succeq 0.$$
So we write the conditions $A\succeq0$ and $\vec{A}^T\vec{A}-2A\bullet P + P\bullet P\leq t$ as the combined condition
$$\begin{bmatrix}
A&0&0\\
0&I&\vec{A}\\
0&\vec{A}^T&2A\bullet P-P\bullet P+t
\end{bmatrix}\succeq 0.$$
The conditions are now
\begin{equation*}
\begin{aligned}
    A=\operatorname{argmin}&&&t\\
    \text{s.t.}&&&A_{ii}=1,\:\forall i=1,\ldots, n,\\
    &&&A_{ij}=0\text{ if }(i,j)\in\overline{E},\\
    &&&\begin{bmatrix}
A&0&0\\
0&I&\vec{A}\\
0&\vec{A}^T&2A\bullet P-P\bullet P+t
\end{bmatrix}\succeq 0.
\end{aligned}
\end{equation*}

Since we are minimizing $t$, and $\tr(PP)$ is a constant then we may omit it to obtain \begin{equation*}
\begin{aligned}
    A=\operatorname{argmin}&&&t\\
    \text{s.t.}&&&A_{ii}=1,\:\forall i=1,\ldots, n,\\
    &&&A_{ij}=0\text{ if }(i,j)\in\overline{E},\\
    &&&\begin{bmatrix}
A&0&0\\
0&I&\vec{A}\\
0&\vec{A}^T&2A\bullet P+t
\end{bmatrix}\succeq 0.
\end{aligned}
\end{equation*} This finishes the proof. \qedhere

\end{proof}

\begin{lem}\label{prop:phi-lin-sdp}
The dual to $\phi_{\Lin}(\Gamma)$ is the program
\begin{equation}\label{eq:phi-lin-dual}
    \begin{aligned}
        \tilde{\phi}_{\Lin}(\Gamma)=\min&&&\tr(Y)\\
        \textup{s.t.}&&&Y_{ij}=0\textup{ if }(i,j)\in E,\\
        &&&Y\succeq J_n/n.
    \end{aligned}
\end{equation}
\end{lem}

\begin{proof}
In the program \eqref{eq:PhiLinGraph}, every constraint is written in the form $E_{ij}\bullet A = \de_{ij}$ and thus we have one constraint on every element $A_{ij}$ except for where $(i,j)\in E$. Therefore the dual program can be written as
\begin{equation*}
    \begin{aligned}
        \tilde{\phi}_{\Lin}(\Gamma)=\min&&&Y\bullet \delta\\
        \text{s.t.}&&&\sum_{(i,j)\notin E}Y_{ij} - J_n/n\succeq 0.
    \end{aligned}
\end{equation*}
where $\delta = [\delta_{ij}]$. Thus, $\delta = I_n$ and $Y\bullet I_n=\tr(Y)$, and we can equivalently write the constraint as
\begin{equation*}
    \begin{aligned}
        \tilde{\phi}_{\Lin}(\Gamma)=\min&&&Y\bullet \de\\
        \text{s.t.}&&&Y_{ij}=0\text{ if }(i,j)\in E,\\
        &&&Y - J_n/n\succeq 0.
    \end{aligned}
\end{equation*}
Since $I_n$ is positive definite and satisfies all constraints of $\phi_{\Lin}$ we conclude $\tilde{\phi}_{\Lin} = \phi_{\Lin}$.
\end{proof}

\begin{prop}\label{eq:phi-lin-product}
For any two graphs $\Gamma=(V_\Gamma,E_\Gamma)$ and $\Lambda=(V_\Lambda,E_\Lambda)$, we have
\begin{equation}
    \phi_{\Lin}(\Gamma\boxtimes\Lambda) = \phi_{\Lin}(\Gamma)\phi_{\Lin}(\Lambda).
\end{equation}
\end{prop}

\begin{proof}
$(\geq)$
We will show that the RHS satisfies the constraints of the LHS.
Let $|V_\Gamma|=n$ and $|V_\Lambda|=k$.
The RHS can be written as
\begin{equation*}
\begin{aligned}
    \phi_{\Lin}(\Gamma)\phi_{\Lin}(\Lambda)=\max&&&\frac{1}{nk}(A_\Gamma\bullet J_n)(A_\Lambda\bullet J_k)\\
    \text{s.t.}&&&[A_\Gamma]_{ii}=1,\:\forall i=1,\ldots, n,\\
    &&&[A_\Lambda]_{\ell\ell}=1,\:\forall \ell=1,\ldots, k,\\
    &&&[A_\Gamma]_{ij}=0\text{ if }(i,j)\in\overline{E_\Gamma},\\
    &&&[A_\Lambda]_{m\ell}=0\text{ if }(m,\ell)\in\overline{E_\Lambda},\\
    &&&A_\Gamma\in M_n^+,\quad A_\Lambda\in M_k^+.
\end{aligned}
\end{equation*}
$A_\Gamma\in M_n^+$ and $A_\Lambda\in M_k^+$ implies that $A_\Gamma\otimes A_\Lambda\in M_{nk}^+$, and $[A_\Gamma]_{ii}=[A_\Lambda]_{\ell\ell}=1$ is equivalent to $[A_\Gamma\otimes A_\Lambda]_{ss}=1$ for $s=1,\ldots, nk$.
Additionally, $(A_\Gamma\bullet J_n)(A_\Lambda\bullet J_k)=(A_\Gamma\otimes A_\Lambda)\bullet(J_n\otimes J_k)=(A_\Gamma\otimes A_\Lambda)\bullet J_{nk}$.

The tensor product of the matricial systems is $S_{\Gamma}\otimes S_{\Lambda}=S_{\Gamma\boxtimes \Lambda}$. Since $A_{\Gamma}\in S_{\Gamma}$ and $A_{\Lambda}\in S_{\Lambda}$ the tensor product $A_\Gamma\otimes A_\Lambda\in S_{\Gamma\boxtimes\Lambda}$, so it satisfies all the conditions for $\phi_{\Lin}(\Gamma\boxtimes\Lambda)$. Therefore $\phi_{\Lin}(\Gamma\boxtimes\Lambda) \geq \phi_{\Lin}(\Gamma)\phi_{\Lin}(\Lambda)$.

$(\leq)$ Starting with the dual program \eqref{eq:phi-lin-dual}, the first condition is equivalent to saying that $Y\in S_{\overline{\Gamma}}$.
Once again, we write the right hand side as
\begin{equation*}
\begin{aligned}
    \tilde{\phi}_{\Lin}(\Gamma)\tilde{\phi}_{\Lin}(\Lambda)=\min&&&(Y_\Gamma\bullet I_n)(Y_\Lambda\bullet I_k)\\
    \text{s.t.}&&&[Y_\Gamma]_{ij}=0\text{ if }(i,j)\in E_\Gamma,\\
    &&&[Y_\Lambda]_{k\ell}=0\text{ if }(k,\ell)\in E_\Lambda,\\
    &&&Y_\Gamma\succeq J_n/n,\quad Y_\Lambda\succeq J_k/k.
\end{aligned}
\end{equation*}
By the same reasoning as in the first direction, we have $Y_\Gamma\otimes Y_\Lambda\in S_{\overline{\Gamma\boxtimes\Lambda}}$.

It remains to show that $Y_\Gamma\otimes Y_\Lambda-J_{nk}/nk\succeq 0$.
Since $J_n\succeq 0$ for any $n$, this yields the relations $J_n/n\otimes(Y_\Lambda-J_k/k)\succeq 0$ and $(Y_\Gamma-J_n/n)\otimes J_k/k\succeq 0$, which is equivalent to $J_n/n\otimes Y_\Lambda\succeq J_n/n\otimes J_k/k$ and $Y_\Gamma\otimes J_k/k\succeq J_n/n\otimes J_k/k$.
Since $Y_\Gamma-J_n/n\succeq 0$ and $Y_\Lambda-J_k/k\succeq 0$, we get that the tensor product $(Y_\Gamma-J_n/n)\otimes(Y_\Lambda-J_k/k)\succeq 0$.
Expanding out the product and using the relations above gives
\begin{align*}
    Y_\Gamma\otimes Y_\Lambda&\succeq J_n/n\otimes Y_\Lambda + Y_\Gamma\otimes J_k/k - J_n/n\otimes J_k/k\\
    &\succeq J_n/n\otimes Y_\Lambda\\
    &\succeq J_n/n\otimes J_k/k.
\end{align*}
By transitivity, this implies $Y_\G\otimes Y_\Lambda\succeq J_n/n\otimes J_k/k=J_{nk}/nk$. $Y_\G\otimes Y_\Lambda$ satisfies every condition of $\tilde{\phi}_{\Lin}(\Gamma\boxtimes \Lambda)$, so $\tilde{\phi}_{\Lin}(\Gamma\boxtimes \Lambda)\leq \tilde{\phi}_{\Lin}(\Gamma)\tilde{\phi}_{\Lin}(\Lambda)$, and by duality $\phi_{\Lin}(\Gamma\boxtimes \Lambda)\leq \phi_{\Lin}(\Gamma)\phi_{\Lin}(\Lambda)$.
\end{proof}

\begin{prop}\label{prop:alt-phi-lin}
For any graph $\Gamma=(V,E)$, $\phi_{\Lin}$ can be written as the vector program
\begin{equation}\label{eq:lin_vector}
        \begin{aligned}
        {1}/{\sqrt{\phi_{\Lin}(\overline{\G})}} = \max&&&t\\
        \textup{s.t.}&&&\mathbf{u}_i^T\mathbf{u}_j = 0\textup{ if }(i,j)\in\overline{E},\\
        &&&\mathbf{c}^T\mathbf{u}_i\geq t,\:\forall i\in V,\\
        &&&\sum_{i=1}^{n}\mathbf{u}_i^T\mathbf{u}_i = n,\\
        &&&\| \mathbf{c}\| = 1.
    \end{aligned}
    \end{equation}
\end{prop}
\begin{proof}
Using the dual (\ref{eq:phi-lin-dual}), make the substitution $X=nY-J_n$:
\begin{equation}
    \begin{aligned}
        \tilde{\phi}_{\Lin}(\Gamma)=\min&&&\tr(X)/n+1\\
        \text{s.t.}&&&X_{ij}=-1\text{ if }(i,j)\in E,\\
        &&&X\succeq 0.
    \end{aligned}
\end{equation}
We will use a proof here similar to \cite{GartnerMatousek2012}*{Thm 3.6.1}. Write the value of the above vector program as $\phi'_{\Lin}(\G)$, and denote $\tilde{t}={1}/{\sqrt{\phi'_{\Lin}(\overline{\G})}}$.

     $(\phi'_{\Lin}(\overline\G)\geq\tilde{\phi}_{\Lin}(\overline\G))$: First, let $\mathcal{U}=(\mathbf{u}_1,\ldots, \mathbf{u}_n)$ be an optimal \textit{orthogonal} representation with handle $\mathbf{c}$. Say that the vectors are chosen such that for some $k$, $\mathbf{c}^T\mathbf{u}_k=t$. Then define the matrix $\tilde X$ with elements
    $$\tilde X_{ij} = \frac{\mathbf{u}_i^T\mathbf{u}_j}{(\mathbf{c}^T\mathbf{u}_i)(\mathbf{c}^T\mathbf{u}_j)}-1=\left(\mathbf{c}-\frac{\mathbf{u}_i}{\mathbf{c}^T\mathbf{u}_i}\right)^T\left(\mathbf{c}-\frac{\mathbf{u}_j}{\mathbf{c}^T\mathbf{u}_j}\right),\quad i\neq j;$$
    $$\tilde X_{ii} = \frac{\mathbf{u}_i^T\mathbf{u}_i}{(\mathbf{c}^T\mathbf{u}_k)^2}-1.$$
    To see that $\tilde{X}$ is positive semidefinite, note that
    $$\tilde X_{ii} \geq \frac{\mathbf{u}_i^T\mathbf{u}_i}{(\mathbf{c}^T\mathbf{u}_i)^2}-1,$$
    Therefore $\tilde{X}=D+S^TS$ for a matrix $S$ and a diagonal matrix $D$ with all nonnegative values. The sum of PSD matrices is PSD, so $\tilde{X}\succeq 0$.
    Note that for $(i,j)\in\overline E$, we assume $\mathbf{u}_i^T\mathbf{u}_j=0$, and thus $\tilde X_{ij}=-1$ as well. To examine the trace, note that
    $$\tr (\tilde X)/n+1 = \frac{1}{n}\sum_{i=1}^n\frac{\mathbf{u}_i^T\mathbf{u}_i}{(\mathbf{c}^T\mathbf{u}_k)^2}=\phi'_{\Lin}(\overline\G).$$
    $\tilde{X}$ satisfies all constraints and therefore $\phi'_{\Lin}(\overline\G)\geq \tilde{\phi}_{\Lin}(\overline\G)$.
    
     $(\phi'_{\Lin}(\overline\G)\leq\tilde{\phi}_{\Lin}(\overline\G))$: Given an optimal argument $X$ to $\tilde{\phi}_{\Lin}(\overline\G)$, let $X=S^TS$ by Cholesky decomposition, with $\mathbf{s}_i$ the columns of $S$. We want to construct a solution to the vector program with value $\tilde{t}$. $X$ is singular; therefore, $S$ must also be singular. Hence the $\mathbf{s}_i$ vectors do not span all of $\mathbb{R}^n$, and there exists a $\mathbf{c}$ orthogonal to all the $\mathbf{s}_i$. Now define $\mathbf{u}_i:=\frac{1}{\sqrt{t}}(\mathbf{c}+\mathbf{s}_i)$, where $t=\tr(X)/n+1$. Note that
    $$\sum_{i=1}^n\mathbf{u}_i^T\mathbf{u}_i = \sum_{i=1}^n \frac{1}{t}(1+X_{ii}) = \frac{1}{t}(n+\tr(X)) = n$$
    as well as the relation that if $(i,j)\in\overline{E}$ then $\mathbf{s}_i^T\mathbf{s}_j=0$. For $(i,j)\in\overline{E}$ we obtain
    $$\mathbf{u}_i^T\mathbf{u}_j = \frac{1}{t}(\mathbf{c}^T\mathbf{c}+\mathbf{c}^T\mathbf{s}_j+\mathbf{s}_i^T\mathbf{c}+\mathbf{s}^T_i\mathbf{s}_j)=\frac{1}{t}(1+X_{ij})=0.$$
    The value of $\tilde{t}$ for the vector program is given by 
    $$\mathbf{c}^T\mathbf{u}_i = \frac{1}{\sqrt{t}}(\mathbf{c}^T\mathbf{c}+\mathbf{c}^T\mathbf{s}_i)=\frac{1}{\sqrt{t}}=\tilde{t}.$$
    We have satisfied every constraint of the vector program, hence $1/\sqrt{\phi'_{\Lin}(\overline\G)}\geq 1/\sqrt{t}$. Therefore $\tilde{\phi}_{\Lin}(\overline\G)\leq t=\phi'_{\Lin}(\overline\G)$.
\end{proof}

\subsection{Counterexamples and Computations}
Here we examine several inequalities which are modifications of the above, yet none of the following are always true for any graph $\G$.

The wheel graph $W_n$ is defined by taking a cycle graph $C_{n-1}$ and connecting each vertex $\{1,\ldots, n-1\}$ to the vertex $n$. The path graph $P_n$ is defined by taking the cycle graph $C_n$ and removing the edge connecting $1$ to $n$. The star graph $S_n$ is defined by connecting each vertex $\{1,\ldots, n-1\}$ and connecting it to vertex $n$, and no other edges are placed. The complete graph is denoted as $K_n$, in which each vertex $\{1,\ldots, n\}$ is connected to every other vertex.

Among some of the identities shown above, there are a few relaxations and related inequalities which are not true.

\begin{prop}\label{thingsthatarewrong}
For any given graph $\Gamma=(V,E)$, the following inequalities are not true:
\begin{enumerate}
    \item \(\phi_{\Lin}(\G)\le \phi_{\Quad}(\G)\);\hspace{\fill} (see Proposition \ref{eq:phi_linquad_ineq})
    
    \item \(\phi_{\Quad}(\G)\ge \omega(\G)\);\hspace{\fill} \quad (see Proposition \ref{prop:phi-quad-omega})
    
    \item \(\phi_{\Quad}(\overline{\G})\ge \vartheta(\G)(\vartheta(\G)-1)/n+1\);\hspace{\fill} (see Proposition \ref{prop:phi-quad-omega})
    
    \item \(\phi_{\Quad}(\G \boxtimes \Lambda)\le\phi_{\Quad}(\G)\phi_{\Quad}(\Lambda) \);\hspace{\fill} (see Proposition \ref{eq:phi-lin-product})
    
    
    \item \(\vartheta(\G)\leq\phi_{\Lin}(\overline\G)\).\hspace{\fill} (see Lemma \ref{prop:phi-lin-sdp} and Proposition \ref{prop:alt-phi-lin})
\end{enumerate}
\end{prop}
\begin{proof} The code by which all of the following results were obtained is contained in the GitHub repository \cite{AGKS}.

Counterexample for \(\phi_{\Lin}(\G)\le \phi_{\Quad}(\G)\). Consider \(\G\) to be a path graph with \(5\) vertices, we obtain the following values:
\begin{itemize}
    \item\(\phi_{\Lin}(\G)=1.9798\),
    \item \(\phi_{\Quad}(\G)=1.9593\),
    \item \(\text{error}=0.0205\).
\end{itemize}


Counterexample for \(\phi_{\Quad}(\G)\ge \omega(\G)\). Consider \(\G\) to be a wheel graph of the order \(5\), we obtain the following values:
\begin{itemize}
    \item \(\phi_{\Quad}(\G)=2.9314\),
    \item \(\omega(\G)=3\),
    \item \(\text{error}= 0.0686\).
\end{itemize}

Counterexample for \(\phi_{\Quad}(\overline{\G})\ge \vartheta(\G)(\vartheta(\G)-1)/n+1\). Consider \(\G=K_{11}\) minus one edge, we obtain the following values:
\begin{itemize}
    \item \(\phi_{\Quad}(\overline{\G}) =1.18181791957969\),
    \item \(\vartheta(\G) = 1.999999999999876\),
    \item \(\text{error}=2.62238484927124\cdot 10^{-7}\).
\end{itemize}
The primal-dual gap from computing $\phi_{\Quad}(\overline\G)$ was given as $\delta\phi=1.45\cdot 10^{-11}$, and the gap from computing $\vartheta(\G)$ was $\delta\vartheta=5.55\cdot10^{-14}$. Ignoring machine precision error, the error of the right hand side of the inequality $f(\G)=\vartheta(\G)(\vartheta(\G)-1)/n+1$ is
$$\delta f = \max_{\pm}|f(\vartheta\pm\delta\vartheta)-f(\theta)|=1.5\cdot10^{-14}.$$
The difference in the two sides is therefore large enough that this error cannot be amounted to lack of precision in the solvers. Albeit a small error, it is still significant enough to be a counterexample.

Counterexample for \(\phi_{\Quad}(\G \boxtimes \Lambda)\le\phi_{\Quad}(\G)\phi_{\Quad}(\Lambda)\). Consider \(\G\) and \(\Lambda\) to be a path graph of five edges, we obtain the following values:
\begin{itemize}
    \item \(\phi_{\Quad}(\G \boxtimes \Lambda)=3.8660\),
    \item \(\phi_{\Quad}(\G)\phi_{\Quad}(\Lambda)=3.8387\),
    \item \(\text{error}=0.0272\).
\end{itemize}

Counterexample for \(\vartheta(\G)\leq\phi_{\Lin}(\overline\G)\). Consider a star graph with \(5\) edges, we obtain the following values:
\begin{itemize}
    \item \(\phi_{\Lin}(\G)=1.8000\),
    \item \(\vartheta(\G)=4.0000\),
    \item \(\text{error}=2.200\).
\end{itemize} 
In this way we have provided counterexamples to all five assertions.\qedhere
\end{proof}

\subsection{Further Discussion and Open Problems}

\begin{question}
    Is $\phi_{\Quad}(\cc S)$ expressible as a semidefinite program?
\end{question}

\begin{question}
Do the following identities hold for all cycle graphs $C_n$?
    \begin{equation}
        \phi_{\Lin}(\overline{C_n})=\vartheta(C_n),\quad \phi_{\Lin}(C_n)=\vartheta(\overline{C_n}).
    \end{equation}
\end{question}

\begin{remark}
The equality $\phi_{\Quad}(\overline{C_n})=\vartheta(C_n)$ does not appear to hold unless $n\leq 5$ or $n=7$.
\end{remark}

\begin{question}
    Is it true that $\phi_{\rm quad}$ or $\phi_{\rm lin}$ are supermultiplicative? That is, is \[\phi_{\rm quad}(\cc S\otimes \cc T)\geq \phi_{\rm quad}(\cc S)\phi_{\rm quad}(\cc T),\quad
    \phi_{\rm lin}(\cc S\otimes \cc T)\geq \phi_{\rm lin}(\cc S)\phi_{\rm lin}(\cc T)\]
    for all operator systems $\cc S\subset M_n$ and $\cc T\subset M_k$?
\end{question}

\section*{Acknowledgements}
Roy Araiza was partially supported as a J.L. Doob Research Assistant Professor at the University of Illinois at Urbana-Champaign. Griffin, Khilnani, and Sinclair were partially supported by NSF grant DMS-2055155. The authors would like to thank Travis Russell for valuable input on an early draft of the manuscript.

\end{document}